\documentclass[10pt,conference,compsocconf,letterpaper]{IEEEtran}
\usepackage[cmex10]{amsmath}
\usepackage{amssymb}
\usepackage{amsthm}
\usepackage{graphicx}
\usepackage{tikz}
\usepackage[font=footnotesize,caption=false]{subfig}

\newcommand{\comment}[1]{}
\newcommand{\E}[1]{\mathbf{E}\left[#1\right]}

\newcommand{\var}[1]{\mathbf{var}\left[#1\right]}

\newcommand{\HarmonicN}[1]{\mathrm{H_{#1}}}

\newcommand{\lambertW}{\mathcal{W}}

\newcommand{\komentarz} [1]{}

\newtheorem{theorem}{Theorem}
\newtheorem{lemma}{Lemma}
\makeatletter
\dimendef\pl@left=0 \dimendef\pl@down=1 
\dimendef\pl@right=2 \dimendef\pl@temp=3
\def\sob#1#2#3#4#5{
\setbox0\hbox{#1}\setbox1\hbox{$_\mathchar'454$}\setbox2\hbox{p}%
\pl@right=#2\wd0 \advance\pl@right by-#3\wd1
\pl@down=#5\ht1 \advance\pl@down by-#4\ht0
\pl@left=\pl@right \advance\pl@left by\wd1
\pl@temp=-\pl@down \advance\pl@temp by\dp2 \dp1=\pl@temp
\leavevmode\kern\pl@right\lower\pl@down\box1\kern-\pl@left #1}
\def\eob{\sob e{.50}{.35}{0}{.93}}
\makeatother

\def\zbyszek{Zbigniew Go{\l}\eob biewski}
\IEEEoverridecommandlockouts

\title{On spreading rumor in heterogeneous systems
\thanks{
Partially supported by EU within the 7th Framework Program under contract 215270 (FRONTS).
This paper was also supported by funds from Polish Ministry of Science and Higher Education -- grant No.  N~N206~1842~33 (first and fourth author), N~N206~2573~35 (second author).}
}

\author{\IEEEauthorblockN{Jacek Cicho\'{n},
\zbyszek,
Marcin Kardas,
Marek Klonowski,
Filip Zag\'orski
}
\IEEEauthorblockA{Faculty of Fundamental Problems of Technology,\\ 
Wroc{\l}aw University of Science and  Technology, Poland}
E-mail: \{firstname.lastname\}@pwr.wroc.pl}

\begin{document}
\maketitle

\begin{abstract}

In this paper we consider a model of spreading information in heterogeneous systems wherein we have two kinds of objects. Some of them are active and others are passive. Active objects can, if they possess information, share it with an encountered passive object. We focus on a particular case such that  active objects communicate independently with randomly chosen passive objects. 

Such model is motivated by two real-life scenarios. The first one is a very dynamic  system of  mobile devices distributing information among stationary devices. The second is an architecture wherein clients communicate with several servers and can leave some information learnt from other servers.

 The main question we investigate is how many rounds is needed to deliver the information to all objects under the assumption that at the beginning exactly one object has the  information?

In this paper we provide mathematical models of such process and show rigid and very precise mathematical analysis for some special cases  important from practical point of view. Some mathematical results are quite surprising ---  we find relation of investigated process to both  coupon collector's problem as well as the birthday paradox. Additionally, we present simulations for showing behaviour for general parameters. 

\end{abstract}

\section{Introduction}

Let us consider the following problem --- we have a heterogeneous system with two  kinds of objects --- namely, \textit{ passive} and \textit{active} ones. In the system information is spread among both kinds of objects. Initially only one object has the information --- it is a source of information. Active objects initiate communication with passive objects. During communication if one object has information they share it. In effect, when the communication is finished --- both objects have the information. The main problem we investigate in our paper is, how many communication rounds is needed to have the information in all objects? In this paper we consider the basic model --- active objects initiate communication independently with randomly chosen passive objects. 
We see this model as a formal description of several real-live scenarios. In particular we have in mind two kinds of systems.  The first one is a   system of  mobile devices distributing information among stationary devices. Of course the assumption that the mobile device initiates the communication with a random stationary object is far-fetched in many real systems. However, it is justified in very dynamic systems where the communication occurs relatively rarely when compared to movement rate. In terms of random walks on graphs theory we can express formally this condition by demanding the period between communications longer than the mixing time of the walk (i.e. movement of the mobile object).

The second scenario that we have in mind is an architecture wherein clients communicate with several servers and can leave information learnt from other servers. The process is finished when all users are informed about the 
information and moreover the information is accessible from all servers. 

To make the presentation more clear we describe our results  in the language of mobile devices. It is clear, however, that it can be seen also in terms of other models.

\begin{figure*}[!t]
 \tikzstyle{mobile}=[circle,draw=blue!50,fill=blue!20,thick,inner sep=0pt,minimum size=4mm]
 \tikzstyle{mobileinf}=[circle,draw=red!70,fill=red!20,thick,inner sep=0pt,minimum size=4mm]
 \tikzstyle{static}=[rectangle,draw=black!50,fill=black!20,thick,inner sep=0pt,minimum size=4mm]
 \tikzstyle{staticinf}=[rectangle,draw=red!70,fill=red!30,thick,inner sep=0pt,minimum size=4mm]
 \tikzstyle{connection}=[<->]
 \tikzstyle{infection}=[<->,draw=red!70]

\def\drawmobile#1{
  \foreach \x/\s in #1
  {
    \newcount\a\a=\x\advance\a by-3\multiply\a by-40
    \node (mob\x) at (\the\a:2cm) [mobile\s] {\x};
  }}

\def\drawstatic#1,#2,#3\koniec{
  \node (st1) at (-1,0) [static#1] {1};
  \node (st2) at ( 0,0) [static#2] {2};
  \node (st3) at ( 1,0) [static#3] {3};}

  \begin{center}
  \subfloat[initialization]
  {
  \label{fig_model:init}
  \begin{tikzpicture}[scale=0.8]
    \drawmobile{{1/,2/,3/,4/,5/,6/,7/,8/,9/}}
    \drawstatic,inf,\koniec
    \draw [connection] (mob9) -- (st3);
  \end{tikzpicture}
  }
  \qquad
  \subfloat[first infection]
  {
  \label{fig_model:inf}
  \begin{tikzpicture}[scale=0.8]
   \drawmobile{{1/inf,2/,3/,4/,5/,6/,7/,8/,9/}}
   \drawstatic ,inf,\koniec
   \draw[infection] (mob1) -- (st2);
  \end{tikzpicture}
  }
  \qquad
  \subfloat[intermediate state]
  {
  \label{fig_model:int}
  \begin{tikzpicture}[scale=0.8]
   \drawmobile{{1/inf,2/,3/inf,4/,5/inf,6/inf,7/,8/,9/inf}}
   \drawstatic inf,inf,\koniec
   \draw[connection] (mob5) -- (st1);
  \end{tikzpicture}
  }
  \qquad
  \subfloat[termination]
  {
  \label{fig_model:term}
  \begin{tikzpicture}[scale=0.8]
   \drawmobile{{1/inf,2/inf,3/inf,4/inf,5/inf,6/inf,7/inf,8/inf,9/inf}}
   \drawstatic inf,inf,inf,\koniec
  \end{tikzpicture}
  }
  \end{center}
  \caption{Example of a rumor spreading process in various states. Mobile and static objects are represented by circles and squares respectively. 
(a) At the beginning,
one object knows the information (static object marked by red). Nodes communicates randomly with no information transmission until some mobile object connects with
infected static object (b). 
From know on, nodes randomly exchanges the rumor (c).
The process stops when all of the objects know the information (d)
}
\end{figure*}
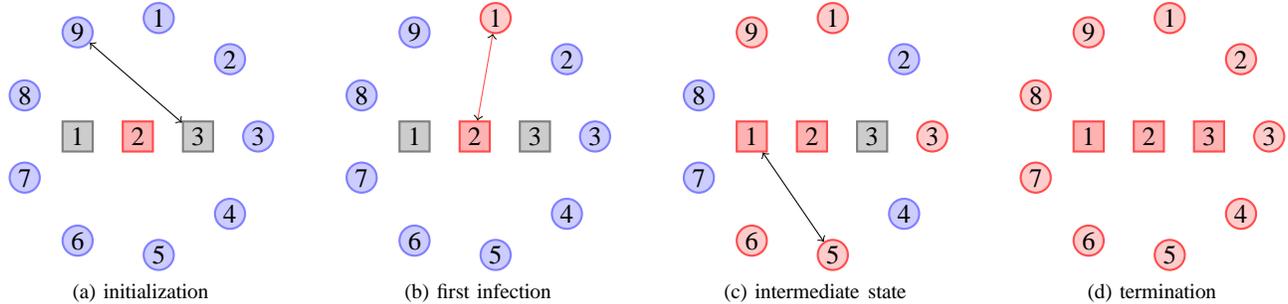

The mathematical model of described system can be a bipartite graph $S\cup M$.
At the beginning, exactly one node from the set $S$ (stationary objects) contains some information.
At each step we choose uniformly at random one node
from the set $S$ and one from the set $M$ (mobile objects). 
If one of the chosen nodes is informed then the second becomes informed too. 
Thus, the considered model is a kind of a \textit{push and pull} rumor spreading scheme.
We are interested in the expected time of total spreading of the information.
We shall show  a very precise result for the case $|M|\leq 3$ and $|S|=n$ and a quite precise result 
for the case $|M| = |S| = n$. 
We also discuss natural continuous analogue of this process as well as some motivation for heterogeneous networks containing some number of  mobile objects.

\subsection{Organization of this paper}

We briefly describe related results --- with particular attention to rumor spreading 
--- in section~\ref{sec:related}. Mathematical model is presented in section~\ref{sec:model}.
Next, in sections~\ref{sec:2bees} and \ref{sec:3bees}, we analyze the cases with small number of mobile objects and we show relation of investigated problem to both birthday paradox and coupon collector's problem.
Section~\ref{sec:equal} is devoted to the case with equal number of clients and servers. Some intuitions about the general case and results of simulations are presented in section\ref{sec:general}. Continuous version of our model is discussed in section~\ref{sec:distributed}.

\section{Related work}
\label{sec:related}
Rumor spreading  has a rich literature. Let us describe one important in this domain result.
In Frieze and Grimmet~\cite{FriezeGrimmet85} it is assumed that
the rumor is initially transmitted to one member of the population and in each round 
each person that knows  the rumor  
passes it to someone chosen at random (it is the so-called {\it push} scheme in a complete graph of size $n$).
The process stops when all people know the rumor. Let $S_n$ 
denote the number of rounds after which the process stops.
It is showed in~\cite{FriezeGrimmet85} that in probability
$$
  S_n = \log_2 n + \ln n + o(\ln n)
$$
when $n$ tends to infinity (i.~e. $\lim_{n \to \infty} \Pr(|S_n - \log_2 n + \ln n| \leq \omega(n)) = 1$ where $\omega(n) = o(\ln n)$).
This result was later improved by Pittel \cite{Pittel87} who showed that in probability
$$
  S_n = \log_2 n + \ln n + O(1) 
$$
when $n$ tends to infinity (i.~e. $\lim_{n \to \infty} \Pr(|S_n - \log_2 n + \ln n| \leq \omega(n)) = 1$ where $\omega(n) = O(1)$). Notice that
\begin{equation}
\label{eq:Pittel}
\log_2 n + \ln n = \left(1+\frac{1}{\ln 2}\right) \ln n \approx 2.4427 \ln n.
\end{equation}

A further idea introduced by Demers et al. in \cite{DBLP:conf/podc/DemersGHIL87} is to send rumors from the called to the calling person ({\it pull} scheme in a complete graph of size $n$).
It was observed that the number of uninformed people decreases much faster using a {\it pull} scheme instead of a {\it push} scheme, namely only $T_n = \Theta(n \ln \ln n)$ calls (transmissions) are needed to distribute the rumor among all players.

In the next crucial paper \cite{DBLP:conf/focs/KarpSSV00} Karp et al. analyzed gossiping protocol which can be briefly described as a combination of {\it push and pull} schemes in a complete graphs of size $n$.
Authors showed that  $S_n = O(\ln n)$.

There are many other important papers in this subject. Let us only mention about
\cite{DBLP:journals/rsa/FeigePRU90},
\cite{DBLP:conf/spaa/Elsasser06},
\cite{DBLP:conf/podc/BerenbrinkEF08},
\cite{1347167}.

Bipartite graphs are suitable to model common client-server architecture.
Moreover, e.g. a social networking websites scheme is a great example of the system where only ``clients'' spread a message among ``servers'' (and therefore also other ``clients'').

\section{Description of Model}
\label{sec:model}
Let $G$ = $(S \cup M,E)$ where $E=S \times M$. 
We call the vertices from the set $S$  \textit{static nodes} 
and we call the vertices from $M$ the  \textit{mobile nodes}.
The only possible way of information exchange is between mobile and static ones ({\it push}\&{\it pull} scheme).
Suppose that one static node has an information item which should be delivered to
all other (mobile and static) nodes. Hence we are interested in a broadcasting 
problem for networks modeled as bipartite graph.
At each step one static and one mobile node is chosen randomly and independently 
according to uniform distributions on $S$ and $M$. If one of them has the information item
then the second will obtain it too.

Let us assume for a while that there is only one mobile node and $n$ static nodes.
Then the evolution of the system
is divided into two phases: time before the first mobile object hits the node 
with the information item (which takes on average $n$ steps) and the dissemination information phase when the mobile
node infects the whole population of static nodes. 
The analysis of the second
phase reduces to the analysis of the Coupon Collector's problem, so the second phase takes in average
$n \HarmonicN{n-1}$ 
steps.  Hence, if $T_n$ is the total number of steps until
the information item is disseminated onto the whole graph, then $\E{T_n} = n\HarmonicN{n-1} + n$.

Let $n$ denote the number of static objects and 
let $m$ denote  the number of mobile objects.
We model our process as a Markov chain with the state space $\Omega^{0}_{n,m} = \{0,\ldots,n\}\times \{0,\ldots,m\}$
and the initial state $(1,0)$.
There are only three possible transitions in this model: 
$(a,b) \rightarrow (a,b)$, $(a,b) \rightarrow (a+1,b)$ and
$(a,b) \rightarrow (a,b+1)$ and we have
$$
\left\{
  \begin{array}{lcl}
  \Pr[(a,b) \rightarrow (a+1,b)] &=&  \frac{b}{m}\cdot \frac{n-a}{n}\\  
  \Pr[(a,b) \rightarrow (a,b+1)] &=&  \frac{m-b}{m}\cdot \frac{a}{n}
  \end{array}
\right.  
$$
and
$\Pr[(a,b) \rightarrow (a,b)] = 1- (\Pr[(a,b) \rightarrow (a+1,b)] + \Pr[(a,b) \rightarrow (a,b+1)])$.

Notice that $\Pr[(1,0) \rightarrow (2,0)] = 0$ and $\Pr[(1,0) \rightarrow (1,1)] = \frac{1}{n}$,
so the only possible evolution of this Markov chain is through the transition $(1,0) \rightarrow (1,1)$.
Therefore we can distinguish two phases in our process:
\begin{description}
 \item [Phase 1]: until the moment of transition from state $(1,0)$ to state $(1,1)$,
 \item [Phase 2]: from state $(1,1)$ until state $(n,m)$ 
\end{description}
The lengths of these two phases are stochastically independent.
Let $F_{n,m}$ be a random variable that measures 
the run time of Phase~1, $L_{n,m}$ be a random variable that measures 
the run time of Phase~2 and let $T_{n,m} = F_{n,m} + L_{n,m}$. 
Then $F_{n,m} \thicksim Geo\left(\frac1n\right)$.
Thus $\E{F_{n,m}} = n$ and $\var{F_{n,m}} = n^2 \left(1 - \frac1n\right)$.
Now we may restrict our considerations to the variable $L_{n,m}$ which 
counts the length of a run through the Markov chain with the state space $\Omega_{n,m} = \{1,\ldots,n\}\times \{1,\ldots,m\}$ 
and the initial state $(1,1)$.

Notice that $\Pr_{n,m}((a,b) \rightarrow (a+1,b)) = \Pr_{m,n}((b,a) \rightarrow (b,a+1))$. 
This observation implies that $\E{L_{n,m}} = \E{L_{m,n}}$ (see Figure~\ref{fig_symmetry}).
\begin{figure}
 \tikzstyle{mobile}=[circle,draw=blue!50,fill=blue!20,thick,inner sep=0pt,minimum size=4mm]
 \tikzstyle{mobileinf}=[circle,draw=red!70,fill=red!20,thick,inner sep=0pt,minimum size=4mm]
 \tikzstyle{static}=[rectangle,draw=black!50,fill=black!20,thick,inner sep=0pt,minimum size=4mm]
 \tikzstyle{staticinf}=[rectangle,draw=red!70,fill=red!30,thick,inner sep=0pt,minimum size=4mm]
 \tikzstyle{connection}=[<->]
 \tikzstyle{infection}=[<->,draw=red!70]

\def\drawobjects#1#2#3#4#5{
\begin{tikzpicture}[scale=0.8]
 \foreach \x/\s in {1/inf,2/,3/inf,4/inf,5/}
   {
    \newcount\a\a=\x\multiply\a by-72\advance\a by162\advance\a by#5
    \node (#1\x) at (\the\a:#3) [#1\s] {\x};
   }
  \foreach \x/\s in {1/,2/inf,3/inf}
  {
   \newcount\a\a=\x\multiply\a by-120\advance\a by30\advance\a by#5
   \node (#2\x) at (\the\a:#4) [#2\s] {\x};
  }
  \draw[connection] (#15) edge [bend right] (#21);
\end{tikzpicture}}

 \begin{center}
 \subfloat[]
 {
  \drawobjects{mobile}{static}{2cm}{0.8cm}{0}
 \label{fig_symmetry1}
 }
 \qquad
 \subfloat[]
 {
  \label{fig_symmetry2}
  \drawobjects{static}{mobile}{0.8cm}{2cm}{180}
 }
 \end{center}
  \caption{Symmetry of rumor spreading process. From state $(1,1)$ (i.e. after initialization phase) process with $n$ static and $m$ mobile objects behaves in the same fashion
as symmetrical process with $m$ static and $n$ mobile objects. 
(a) Process with 5 mobile and 3 static objects in state $(2,3)$ and 
(b) corresponding process
with 3 mobile and 5 static objects in state $(3,2)$. Both marked connections (causing no information exchange) have the same probability of being chosen.}
 \label{fig_symmetry}
\end{figure}
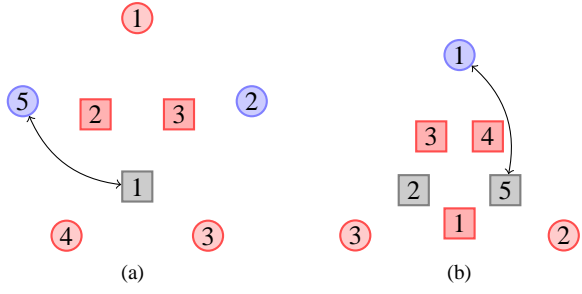


\section{Two Mobile Objects}
\label{sec:2bees}
In this section we assume there are precisely two mobile
stations, i.e. investigate the Markov chain $\Omega_{n,2}$.
In order to simplify notation we write $T_n$ instead of $T_{n,m}$ etc.

Quite easy arguments can show that 
$n + n \HarmonicN{n} \leq \E{T_n} \leq \frac 32 n + n \HarmonicN{n}$.
We shall show in this section a much more precise estimation of $\E{T_n}$.
\bigskip

The probabilities of transitions between states in the considered Markov chain are given by formulas
$$
\left\{
  \begin{array}{lcl}
  \Pr[(a,1) \rightarrow (a+1,1)] &=& \frac{n-a}{2n}\\
  \Pr[(a,1) \rightarrow (a,2)]   &=& \frac{a}{2n}\\
  \Pr[(a,1) \rightarrow (a,1)] &=& \frac{1}{2}
  \end{array}
\right.
$$

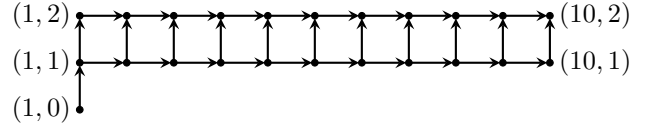
\begin{figure}[!t]
  \begin{center}
    \begin{tikzpicture}[scale=0.625]
  \filldraw (0,0) circle (2pt);
  \tikzstyle{myarrow}=[->,>=stealth,shorten >=1pt,thick];
  \foreach \x in {0,...,10}
  {
    \filldraw (\x,1) circle (2pt);
    \filldraw (\x,2) circle (2pt);
    \begin{scope}[style=myarrow]
      \draw (\x,1) -- +(0,1);
    \end{scope}
  }
  \foreach \x in {0,...,9}
  {
    \draw[style=myarrow] (\x,1) -- +(1,0);
    \draw[style=myarrow] (\x,2) -- +(1,0);
  }
  \draw[style=myarrow] (0,0) -- (0,1);
  \draw (0,0) node [left] {$(1,0)$};
  \draw (0,1) node [left] {$(1,1)$};
  \draw (0,2) node [left] {$(1,2)$};
  \draw (10,1) node [right] {$(10,1)$};
  \draw (10,2) node [right] {$(10,2)$};
\end{tikzpicture}
  \end{center}
  \caption{States of the model with two mobile objects and 10 static objects. Loops are omitted.}
  \label{fig:2bees_model}
\end{figure}

We restrict for a while our attention only to the moments, when the evolution of our Markov chain occurs (i.~e. we consider the  pure jump Markov chain).
For the resulting pure jump Markov chain we have
$$
\left\{
  \begin{array}{lcl}
    \tilde{\Pr}[(a,1) \rightarrow (a+1,1)] &=& 1-\frac{a}{n}\\
    \tilde{\Pr}[(a,1) \rightarrow (a,2)]   &=& \frac{a}{n}
  \end{array}
\right.  
$$

Let $p_a$ 
denote 
 the probability of the evolution through the path 
\begin{equation}
  \label{eq:2bees:path}
  (1,1)\rightarrow \ldots \rightarrow (a,1) \rightarrow (a,2) \rightarrow (a+1,2) \rightarrow \ldots \rightarrow (n,2)
\end{equation}
Then $p_1 = \frac{1}{n}$ and
\begin{equation}
  p_a = \frac{a}{n} \prod_{k=1}^{a-1} \left(1- \frac{k}{n}\right) \label{pa}
\end{equation}
for $a>1$. Let us fix a parameter $a$ and let us calculate the expected 
number 
of steps ($L_{n}^{(a)}$)
which will take in the original chain to evolve through the path given by Formula~\ref{eq:2bees:path}.
Notice that $\Pr[(a,1) \rightarrow (a,1)] = \frac{1}{2}$
and $\Pr[(a,2) \rightarrow (a,2)] = \frac{a}{n}$.
Therefore we have
$$
  \E{L_{n}^{(a)}} = 2 a + \sum_{k=a}^{n-1} \frac{1}{1- \frac{k}{n}} = 2 a + n \HarmonicN{n-a} .
$$
The expected number of steps in Phase 2 before reaching the final absorbing state $(n,2)$
is given by formula $\E{L_n} = \sum_{a=1}^{n} p_a \E{L_{n}^{(a)}}$.
Let us recall that $(p_a)$ is a probability distribution.
Taking into account the initial numbers of steps which are necessary to
make the jump $(1,0)\rightarrow(1,1)$ we get 
\begin{equation}
\label{eq:exp2bees} 
  \E{T_n} = n + 2 \sum_{a=1}^{n} a p_a + n \sum_{a=1}^{n-1} p_a \HarmonicN{n-a}.
\end{equation}

\subsection{Expected value \label{sec:2pszczoly:exp}}
\label{sec:basic:function}

Let us consider for a while the classical urns and balls model and let $B_n$ 
denote the moment of the first
collision where $n$ denotes the number of urns. Then $\Pr[B_n = 0] = 0$ and $\Pr[B_n = a] = p_{a-1}$ for $a>0$.
Therefore we see that 
from state $(1,1)$ up to state $(a,2)$
our process is closely 
related to 
the Birthday Paradox, namely
$\sum_{a=1}^{n} a p_a + 1 = \E{B_n}$.
Then, from state $(a,2)$ to state $(n,2)$ (the end of rumor spreading) the process acts like partial Coupon Collector's Problem (i.~e. collecting $n-a$ from $n$ coupons).
From this considerations and a precise formula for $\E{B_n}$  (see e.g. \cite{ANACOMB}) we get
\begin{equation}
  \label{eq:sum:1}
  \sum_{a=1}^{n} a p_a = \sqrt{\frac{\pi n}{2}} - \frac13 + \frac{1}{12}\sqrt{\frac{\pi}{2n}} - \frac{4}{135n} + O\left(n^{-\frac32}\right).
\end{equation}

Now we shall find an asymptotic expression for the sum   $\sum_{a=1}^{n-1} p_a \HarmonicN{n-a}$
with a precision of order $O(n^{-\alpha})$ for some $\alpha >1$.
Notice that the sequence $(p_a)$ is unimodal 
and reaches its maximum near $\sqrt{n}$.
Using the simple estimate $1-x < e^{-x}$ we deduce that  $p_{a+1}  <   e^{-\frac{a^2}{2 n}}$
for $1 < a \leq n$.
Therefore, if $a>n^{\frac23}$, then $p_a<e^{-\sqrt[3]{n}}$. We say that $f=O\left(\frac{1}{n^*}\right)$, if
for any $k>0$ we have $f=O(n^{-k})$. Hence, if $f(a,n)$ is a function such that $|f(a,n)|\leq C n^k$
for some fixed $C$ and $k>0$, then
$$
  \sum_{a=n^\frac23+1}^{n}p_a f(a,n) = O\left(\frac{1}{n^*}\right).
$$
We shall denote by $\sum_{a}^*$ the sum $\sum_{a=1}^{\left\lceil n^\frac23\right\rceil}$.
Let us observe that
$$
  \sum_{a=1}^{n-1} p_a \HarmonicN{n-a} = \sum_a^* p_a \HarmonicN{n-a} + R_1(n),
$$
where $R_1 =  O(\frac{1}{n^*})$.
From the formula $\HarmonicN{n} = \ln n + \gamma + \frac{1}{2 n} + O(n^{-2})$ 
we get
\begin{equation}
\label{pszczoly:pahna}
\begin{split}
 \sum_a^* p_a \HarmonicN{n-a} = \gamma &+ \sum_a^* p_a \ln(n-a)\\
 &+ \sum_a^* p_a \frac{1}{2(n-a)} + R_2(n),
\end{split}
\end{equation}
where $R_2 = O(n^{-2})$.\bigskip

Using carefully the expansion $\frac{1}{n-a} = \frac{1}{n} \sum_j (\frac{a}{n})^j$ and the formula
$\ln(n-a) = \ln n - \sum_{j \geq 1} \frac1j(\frac{a}{n})^j$, after some simple but tedious calculus we get
\begin{equation}
\label{eq:sum:2}
\begin{split}
\sum_{a=1}^{n-1} p_a \HarmonicN{n-a} &= \gamma + \ln n - \sqrt{\frac{n \pi}{2}} - \frac{1}{6 n}\\
&\quad - \frac{1}{12 n^{\frac32}} \sqrt{\frac{\pi}{2}}
 + \frac{4}{135 n^2} + O({n^{-\frac73}}).
\end{split}
\end{equation}

\begin{theorem}
\label{thm:expT2n}
Let $T_n$ denote the number of necessary steps to propagate message in a network with $2$ mobile objects and $n$ static objects. Then
\begin{equation}
\label{eq:expT2n}
\begin{split}
 \E{T_n} &= n \HarmonicN{n} + n + \sqrt{\frac{\pi  n}{2}} - \frac{4}{3}\\
&\quad + \frac{1}{12} \sqrt{\frac{\pi}{2 n}} - \frac{4}{135 n} + O\left(n^{-\frac43}\right).
\end{split}
\end{equation}
\end{theorem}
\begin{proof}
  The formula for $\E{T_n}$ follows from Equations \ref{eq:exp2bees}, \ref{eq:sum:1} and \ref{eq:sum:2}.
\end{proof}

\subsection{Birthday paradox and coupon collector's problem}
\label{sec:2bees_relation}
As it was stated in the previous section, the model with two mobile objects is closely related to the birthday paradox
and coupon collector's problem. While this relation is quite understandable for a given path of evolution~\eqref{eq:2bees:path},
it's blurred by the fact that we have to take into account the probabilities $p_a$ of the process evolving through that specified path, as expressed by \eqref{eq:exp2bees}.
However, it can be quantitatively concluded from equation~\eqref{eq:expT2n} that the relation is still present in the process as whole --- namely, that
$$
\E{T_{2,n}}-(\E{F_n}+\E{B_n}+\E{C_n})=-2+O(n^{-\frac43}),
$$
where $F_n,B_n$ and $C_n$ are random variables determining number of steps in the initialization phase,
the birthday paradox and the coupon collector's problem respectively.

To explain the above result, let us refer again to the urns and balls model. In the coupon collector's problem, we start with $n$ empty urns into which
we randomly throw balls until all of the urns contain at least one. While in the classical approach to analysis of the problem (like in~\cite{MMITZENMACHER}), the process
is divided into $n$ phases determined by the number of the non-empty urns, here we want to distinguish only two phases. In the first phase we throw the balls
up to the first collision (or until all of the urns contain exactly one ball) after which the second phase starts and lasts to the end of the process (in particular, it
can happen that the second phase won't occur at all). Let $k$ be the number of steps in the first phase.
If $k=n$, then all of the $n$ urns contains exactly one ball and the process stops without the second phase. For $k<n$ there are exactly $n-k+1$ empty urns which
have to be filled in the second phase. From above considerations we get
\begin{align*}
\E{C_n}&=\sum_{k=2}^n p_{k-1}(k+n\HarmonicN{n-k+1})+p_n n\\
&=\sum_{k=1}^n p_k(k+n\HarmonicN{n-k})+1-p_n.
\end{align*}
By substituting this result to formula~\eqref{eq:exp2bees} we have
\begin{equation}
\label{eq:T2nrel}
 \E{T_{2,n}}=\E{F_n}+\E{B_n}-1+\E{C_{n-1}}+\frac{n!}{n^n},
\end{equation}
getting the exact relation between the considered problems. It should be noticed, that, while the above approach simplifies calculations from the previous subsections,
they are still needed to obtain the variance.

\section{Three mobile objects}
\label{sec:3bees}
The analysis of the case with exactly three mobile objects can be carried using approach analogous to that from section~\ref{sec:2bees}, however it gets much more
complicated. We can divide the process into 3 phases determined by the number of static objects knowing the rumor. We define the probability
$$
 p_{ab} = \frac{2a}{n+a} \frac{b}{2n-b} \prod_{k=1}^{a-1}\frac{n-k}{n+k} \prod_{k=a}^{b-1}\frac{2(n-k)}{2n-k}
$$
of event that the process evolves with transitions $(a,1)\rightarrow(a,2)$ and $(b,2)\rightarrow(b,3)$ (i.e. that the first and second phase ends respectively with exactly
$a$ and $b$ static objects knowing the information, see Figure~\ref{fig:3bees_path}).
\begin{figure}
 \begin{center}

\begin{tikzpicture}[scale=0.65]
 \tikzstyle{myarrow}=[->,>=stealth,thick];
 \draw[gray,densely dotted] (0,0) grid (12,3);
 \draw[myarrow] (0,0) -- ++(0,1);
 \draw[myarrow] (0,1) -- ++(1,0);
 \draw[thick,dashed] (1,1) -- ++(2,0);
 \draw[myarrow] (3,1) -- ++(1,0);
 \draw[myarrow] (4,1) -- ++(0,1);
 \draw[myarrow] (4,2) -- ++(1,0);
 \draw[thick,dashed] (5,2) -- ++(2,0);
 \draw[myarrow] (7,2) -- ++(1,0);
 \draw[myarrow] (8,2) -- ++(0,1);
 \draw[myarrow] (8,3) -- ++(1,0);
 \draw[thick,dashed] (9,3) -- ++(2,0);
 \draw[myarrow] (11,3) -- ++(1,0);
 \draw (0,-1) node [anchor=south] {$1$};
 \draw (1,-1) node [anchor=south] {$2$};
 \draw (2.5,-1) node [anchor=south] {$\cdots$};
 \draw (4,-1) node [anchor=south] {$a$};
 \draw (6,-1) node [anchor=south] {$\cdots$};
 \draw (8,-1) node [anchor=south] {$b$};
 \draw (10,-1) node [anchor=south] {$\cdots$};
 \draw (12,-1) node [anchor=south] {$n$};
 \draw (0,0) node [left] {$0$};
 \draw (0,1) node [left] {$1$};
 \draw (0,2) node [left] {$2$};
 \draw (0,3) node [left] {$3$};
\end{tikzpicture}

 \end{center}
 \caption{Possible path of evolution of the process with 3 mobile objects.}
 \label{fig:3bees_path}
\end{figure}
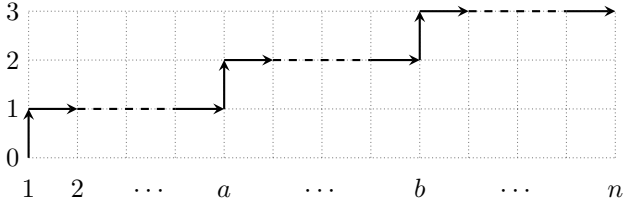

Arguments equivalent to those from previous section lead to expression
\begin{align*}
 \E{T_{3,n}} &= n+\sum_{a=1}^{n}\sum_{b=a}^{n} p_{ab} \Bigg(\sum_{k=1}^{a} \frac{3n}{n+k}\\
&\quad+ \sum_{k=a}^{b} \frac{3n}{2n-k}
+n\HarmonicN{n-b}\Bigg),
\end{align*}
which seems to be much harder to analyze. However, manipulations similar to that from section~\ref{sec:2bees_relation} result in formula analogous to \eqref{eq:T2nrel}
\begin{align*}
 \E{T_{3,n}} &= \E{F_n}+\E{C_{n-1}}+\frac32\E{P^{(1)}_n}\\
&\quad+\frac12 \E{P^{(2)}_n}+\frac{2^{n+3}-1}{3\binom{2n}{n}}~,
\end{align*}
with $F_n$ and $C_n$ being declared as before and $P^{(1)}_n$ and $P^{(2)}_n$ denoting number of static objects knowing the rumor at the end of the first and second phase respectively.
By calculation of $\E{P^{(1)}_n}$ and $\E{P^{(2)}_n}$ we get the following
\begin{theorem}
Let $T_n$ denote the number of necessary steps to propagate message in a network with $3$ mobile and $n$ static objects. Then
\[
\E{T_{n}} = n \HarmonicN{n}+n+\frac{4}{3}\sqrt{\pi n}-\frac{5}{3}+\frac{1}{6}\sqrt{\frac{\pi}{n}}+\frac{1}{96}\sqrt{\frac{\pi}{n^3}}+O(n^{-\frac{5}{2}}).
\]
\end{theorem}


\section{Equal Number of Static and Mobile Nodes} 
\label{sec:equal}
In this section we assume that there are  $n$ static and $n$ mobile
objects, i.e. we deal with the Markov chain with the state space $\Omega_{n,n}$
and the transitions probabilities described in the section \ref{sec:model}.
The presented solution is surprisingly simple.
For $(i,j) \in \Omega_{n,n}$ we put
$p_{i,j} = \Pr[(i,j) \rightarrow (i,j)]$.

\begin{lemma}
\label{lemma:pr_n_n}
  If $1\leq a \leq 2 n$ then\\
  \indent 1. $\max\{p_{i,j}:i+j=a\} = 
     p_{\left\lfloor\frac{a}{2}\right\rfloor,a-\left\lfloor\frac{a}{2}\right\rfloor} =   
     p_{a-\left\lfloor\frac{a}{2}\right\rfloor,\left\lfloor\frac{a}{2}\right\rfloor}$,\\
  \indent 2. $\min\{ p_{i,j}: i+j=a\} = p_{1,a-1} = p_{a-1,1}$. 
\end{lemma}

\begin{figure}[!t]
  \begin{center}
    \begin{tikzpicture}[scale=0.6]
 \draw[gray,densely dotted](0, 0) grid (10,10);
 \foreach \x in {0,...,9}
 {
  \begin{scope}[thick,shorten >=1pt,shorten <=1pt,>=stealth,->]
    \draw [color=blue] (\x,0) -- +(1,0);
    \draw [color=blue] (10,\x) -- +(0,1);
    \draw [color=red ] (\x,\x) -- +(0,1);
    \draw [color=red ] (\x,\x) +(0,1) -- +(1,1);
  \end{scope}
  \filldraw (\x,\x) circle (2pt);
  \filldraw (\x,\x) +(0,1) circle (2pt);
  \filldraw (\x,0) +(1,0) circle (2pt);
  \filldraw (10,\x) +(0,1) circle (2pt);
 }
 \draw (0,0) node [below left] {$(1,1)$};
 \draw (10,10) node [above right] {$(n,n)$};
\end{tikzpicture}
  \end{center}
  \caption{Red arrows shows the slowest path and the blue path shows the fastest path.}
  \label{fig:nxn_model}
\end{figure}
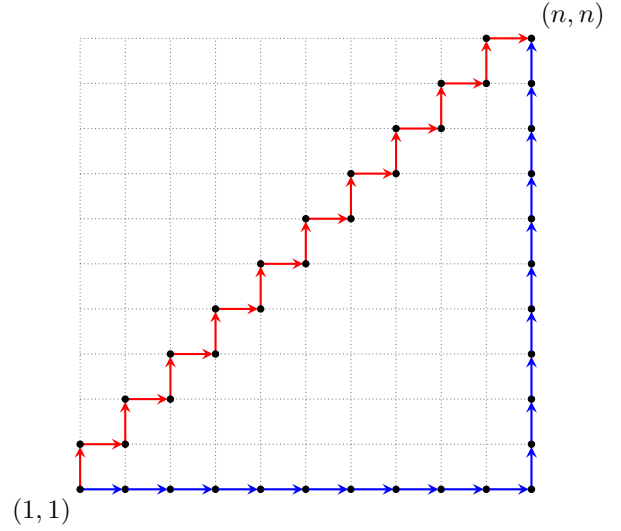

\noindent
We skip a simple proof of  this lemma. 
From this lemma 
we deduce that the probability of changing  state on the line
given by $i+j=a$ is smallest when $|i-j|$ is minimal. 
Therefore every realization of the walk from the state $(1,1)$ to $(n,n)$ 
is stochastically dominated by the length of  the walk restricted to states $\left(\frac{a}{2}, \frac{a}{2}\right)$ if $a$ is even and states $\left(\frac{a+1}{2}, \frac{a-1}{2}\right)$ or 
$\left(\frac{a+1}{2}, \frac{a-1}{2}\right)$ when $a$ is odd.

\begin{theorem}
Let $T_{n,n}$ denote the time necessary to reach the state $(n,n)$ starting from the state $(1,0)$. 
Then 
$$    2 n \HarmonicN{n} + 2\ln n + O(1) \leq \E{T_{n,n}} < 2 n \HarmonicN{n} +  \ln(4) n + O(1).$$
\end{theorem}

\begin{proof}
 Let $S_{i,j}$ denote the number of rounds in which the state $(i,j)$ does not change.
By 
Lemma~\ref{lemma:pr_n_n} (see Figure \ref{fig:nxn_model}) we get 
 $$
   \E{T_{n,n}} \geq n + \sum_{j=1}^{n} \E{S_{1,j}} + \sum_{j=2}^{n-1} \E{S_{j,n}}.
 $$
Random variables $S_{i,j}$ are geometrically distributed, hence
 \begin{gather*}
 \sum_{j=2}^{n-1} \E{S_{j,n}} = \sum_{j=2}^{n-1} \frac{1}{1-p_{j,n}}
 = \sum_{j=2}^{n-1} \frac{1}{1 - \frac{j}{n}}
 = n \HarmonicN{n} + O(1) 
\end{gather*}
and
\begin{gather*}
  \sum_{j=1}^{n} \E{S_{1,j}} = \sum_{j=1}^{n} \frac{1}{1-p_{1,j}}
  = \sum_{j=1}^{n} \frac{1}{\frac1n \left(1-\frac{j}{n}\right) + \frac{j}{n}\left(1-\frac1n\right)}\\
 = \frac{n^2}{n-2} \sum_{j=1}^{n} \frac{1}{j + \frac{n}{n-2}} = 
  n \HarmonicN{n} - n + 2 \ln n + O(1).
 \end{gather*}
Therefore
$$
  \E{T_{n,n}} \geq 2 n \HarmonicN{n} + 2\ln n + O(1).
$$

In order to find an upper bound for $\E{T_{n,n}}$ we shall consider
the slowest  
path of an evolution of the considered Markov chain (see Figure \ref{fig:nxn_model}).
Notice that if $a$ is even, then among all states $(i,j)$ such that
$i+j=a$ we should consider the state $(\frac{a}{2},\frac{a}{2})$. 
If $a$ is odd then 
we have to consider states $(\frac{a+1}{2},\frac{a-1}{2})$ and $(\frac{a-1}{2},\frac{a+1}{2})$,
but the chain is symmetric, so these two cases are equivalent, so we may choose any of them.
Hence 
$$
\E{T_{n,n}} \leq n +  
  \sum_{j = 1}^{n-1} \left(\E{S_{j,j}} + \E{S_{j,j+1}}\right).
$$
But
 \begin{align*}
  \sum_{j=1}^{n-1} \E{S_{j,j}} &= 
  \sum_{j=1}^{n-1} \frac{1}{1 - p_{j,j}}\\
&= 
  \sum_{j=1}^{n-1} \frac{1}{2 \frac{j}{n}    \left(1 - \frac{j}{n}\right)} =  n \HarmonicN{n} - 1.
 \end{align*}
and
\begin{align*}
 \sum_{j=1}^{n-1} \E{S_{j,j+1}} &= 
 \sum_{j=1}^{n-1} \frac{1}{1 - p_{j,j+1}}\\
&= 
 n^2 \sum_{j=1}^{n-1} \frac{1}{-2 j (1+j) + 2jn + n}\\
 &= n (\ln n + \gamma + \ln 4 - 1) + O(1)\\
&= 
 n \HarmonicN{n} + (\ln 4 - 1)n + O(1),
\end{align*}
hence
$\E{T_{n,n}} \leq 2 n \HarmonicN{n} + \ln(4) n + O(1)$.   
\end{proof}

\section{General case}
\label{sec:general}
\begin{figure}
\begin{center}
   \begin{tikzpicture}[scale=1.2]
   \draw[->] (0,0) -- coordinate (x axis mid) (5.2,0);
   \draw[->] (0,0) -- coordinate (y axis mid) (0,5.2);
  \draw (5cm,-3pt) -- (5cm,3pt) node [below=2pt] {\footnotesize{$10^5$}};
  \draw (-3pt,5cm) -- (3pt,5cm) node [right] {\footnotesize{$2.5\cdot 10^6$}};
  \draw (0,0) node [left=3pt,below=0pt] {\footnotesize{0}};
  \node[below=5pt] at (x axis mid) {$n$};
  \node[rotate=90,above=-2pt] at (y axis mid) {\# of steps};
   \draw[xscale=0.05,yscale=0.019] plot file {explin.txt};
   \draw[xscale=0.05,yscale=0.019,red,dashed] plot file {upper.txt};
   \draw[xscale=0.05,yscale=0.019,red,dashed] plot file {lower.txt};
  \end{tikzpicture}
  \caption{Numerically computed expected number of steps in case of equal number of mobile and static objects (black line) and our upper and lower bounds (red dashed lines). Notice, that upper bound line covers plot of the expected value of steps.}
  \label{fig:lin_bounds}
\end{center}
\end{figure}
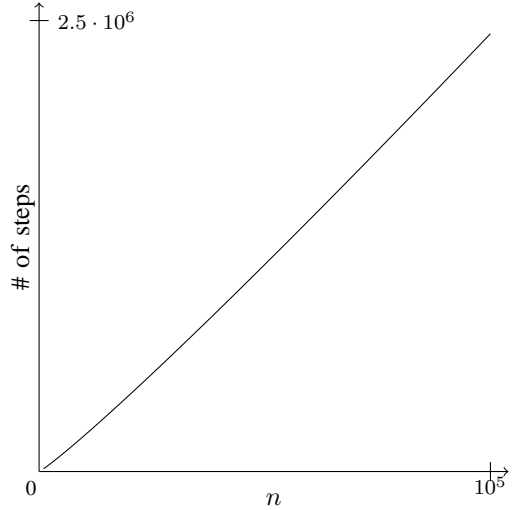

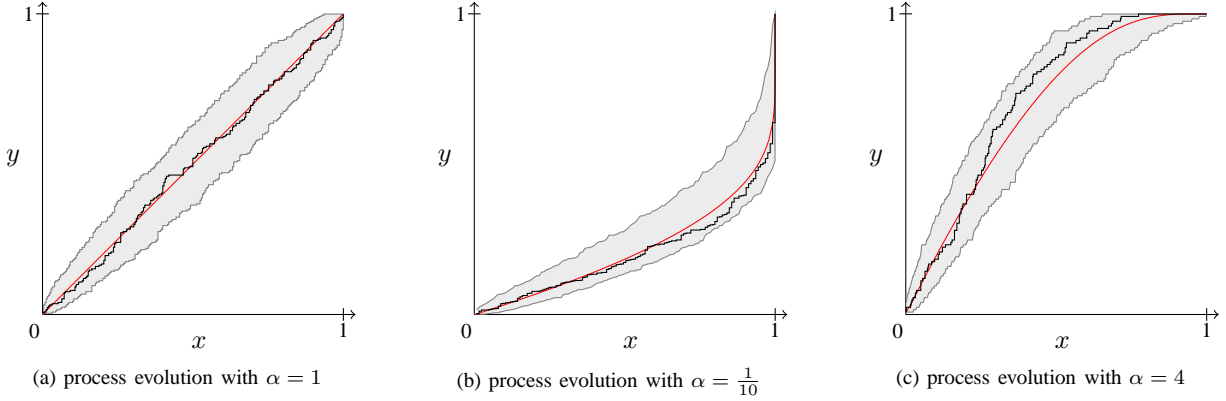
\begin{figure*}[!t]
 \begin{center}
 \subfloat[process evolution with $\alpha=1$]{
  \begin{tikzpicture}[scale=0.8]
  \draw[->] (0,0) -- coordinate (x axis mid) (5.2,0);
  \draw[->] (0,0) -- coordinate (y axis mid) (0,5.2);
  \draw (5cm,-3pt) -- (5cm,3pt) node [below=2pt] {\footnotesize{1}};
  \draw (-3pt,5cm) -- (3pt,5cm) node [left=2pt] {\footnotesize{1}};
  \draw (0,0) node [left=3pt,below=0pt] {\footnotesize{0}};
  \node[below=5pt] at (x axis mid) {$x$};
  \node[left=5pt] at (y axis mid) {$y$};
  \fill[xscale=0.02,yscale=0.02,draw=gray,fill=gray!15] plot file {symbg250x250.txt};
  \draw[red] (0,0) -- (5,5);
  \draw[xscale=0.02,yscale=0.02] plot file {sym250x250.txt};
  \end{tikzpicture}
  \label{fig:experiment_linear}
 }
 \qquad
 \subfloat[process evolution with $\alpha=\frac{1}{10}$]{
  \begin{tikzpicture}[scale=0.8]
  \draw[->] (0,0) -- coordinate (x axis mid) (5.2,0);
  \draw[->] (0,0) -- coordinate (y axis mid) (0,5.2);
  \draw (5cm,-3pt) -- (5cm,3pt) node [below=2pt] {\footnotesize{1}};
  \draw (-3pt,5cm) -- (3pt,5cm) node [left=2pt] {\footnotesize{1}};
  \draw (0,0) node [left=3pt,below=0pt] {\footnotesize{0}};
  \node[below=5pt] at (x axis mid) {$x$};
  \node[left=5pt] at (y axis mid) {$y$};
  \fill[xscale=0.04,yscale=0.05,draw=gray,fill=gray!15] plot file {symbg1000x125.txt};
  \draw[xscale=5,yscale=5,red] plot file {exp10x1.txt};
  \draw[xscale=0.04,yscale=0.05] plot file {sym1000x125.txt};
  \end{tikzpicture}
  \label{fig:experiment_n10}
  }
 \qquad
 \subfloat[process evolution with $\alpha=4$]{
  \begin{tikzpicture}[scale=0.8]
  \draw[->] (0,0) -- coordinate (x axis mid) (5.2,0);
  \draw[->] (0,0) -- coordinate (y axis mid) (0,5.2);
  \draw (5cm,-3pt) -- (5cm,3pt) node [below=2pt] {\footnotesize{1}};
  \draw (-3pt,5cm) -- (3pt,5cm) node [left=2pt] {\footnotesize{1}};
  \draw (0,0) node [left=3pt,below=0pt] {\footnotesize{0}};
  \node[below=5pt] at (x axis mid) {$x$};
  \node[left=5pt] at (y axis mid) {$y$};
  \fill[xscale=0.01,yscale=0.04,draw=gray,fill=gray!15] plot file {symbg125x500.txt};
  \draw[xscale=5,yscale=5,red] plot file {exp1x4.txt};
  \draw[xscale=0.01,yscale=0.04] plot file {sym125x500.txt};
  \end{tikzpicture}
  \label{fig:experiment_4n}
  }
  \end{center}
  \caption{Example realization of process (black line) with accumulated results of $10^5$ simulations (gray area). The process seems to be concentrated around solution (red line) of differential equation~\eqref{eqn:yx}.}
\label{fig:sim}
\end{figure*}

In previous sections we investigated some special yet very important cases of our rumor spreading process. To complete the considerations,
we present here results for the process in general case, mostly obtained by simulations and numerical computations.
We are interested in information spreading between $n$ mobile and $\alpha n$ static objects, where $\alpha$ is proportion coefficient.
To get some intuitions about process behavior, we analyze continuous version of process by using differential equations.
With $x(t)$ and $y(t)$ being proportions of infected objects of static and mobile type respectively, we get
\begin{equation*}
\begin{cases}
 \dot{x}=\frac{1}{\alpha n}(1-x)y\\
 \dot{y}=\frac{1}{n}x(1-y)
\end{cases} 
\end{equation*}
with initial conditions $x(0)=\frac{1}{\alpha n}$ and $y(0)=\frac1n$. Solving this ODE leads to
\begin{equation}
\label{eqn:yx}
y(x)=\lambertW\bigl((\frac1n-1)(\frac{1}{\alpha n}-1)^{-\alpha}e^{\alpha x-1}(x-1)^\alpha\bigr)+1~,
\end{equation}
where $\lambertW(z)$ is the Lambert $\lambertW$-function (the inverse function of $z(w)=w e^w$).
Computer simulations seem to confirm that the evolution paths of process is well approximated by obtained solutions (see Figure~\ref{fig:sim}).

By taking $\alpha=1$ we get the equal-case from section \ref{sec:equal} for which equation~\eqref{eqn:yx} turns into $y(x)=x$ (see Figure~\ref{fig:experiment_linear}).
As this continuous solution corresponds to the slowest path of progression (c.f. Figure~\ref{fig:nxn_model}), it appears that the expected value $\E{T_{n,n}}$ is close to
our upper-bound (Figure~\ref{fig:lin_bounds}).

\section{Distributed model}
\label{sec:distributed}
In previous sections we were considering model with objects communicating in rounds with exactly one connection per round. Such assumption leads to counterintuitive results,
that f.e. disseminating information in network consisted of $n$ clients and one server is expected to end faster than in network of the same size, with $n-1$ clients and two servers. The problem is
that in the described model objects work sequentially and we can't benefit from larger number of servers.

In practice applications, stations participating in communication run in parallel. Let us assume then, that mobile objects are randomly awaken every once in a while.
Precisely speaking, if $N_i(t)$ denotes number of connections established by $i^{th}$ mobile object up to time $t$, then we would like (for each $1\leq i \leq m$ and some parameter $\lambda$)
$N_i(t)$ to be stochastic counting process such that the following statements hold (as described in~\cite{MMITZENMACHER}):
\begin{enumerate}
 \item $N_i(0)=0$,
 \item the process has independent and stationary increments,
 \item the probability of a single event in a short interval $t$ tends to $\lambda t$,
 \item the probability of more than one event in a short interval $t$ tends to $0$.
\end{enumerate}
The above, relatively natural conditions, define $N_i(t)$ to be the Poisson process with parameter $\lambda$. What is interesting here, the sum $N(t)=N_1(t)+\ldots +N_m(t)$ is the Poisson
process with the parameter $m\lambda$, which means that the total number of activations in time unit is proportional to the number of mobile objects. Moreover, the probability of
collision (which has to be considered since stations run in parallel) in a short interval tends to $0$. This lets us to apply results from previous sections to the distributed model.
Namely, suppose 
that all mobile objects are activated independently with exponential distribution with parameter $\lambda$. 
Let $T_{n,m}(\lambda)$ be the time necessary for broadcasting an information by $m$ mobile 
objects in this model when the mobile objects are activated with intensity $\lambda$. 
Then

\begin{align*}
	\E{T_{n,1}(\lambda)} &= \ln (n)/\lambda+ O(n),\\
  \E{T_{n,2}(\lambda)} &= \ln (n)/(2\lambda)+ O(n), \\
  \E{T_{n,3}(\lambda)} &= \ln (n)/(3\lambda)+ O(n), \\
  \E{T_{n,n}(\lambda)} &= 2\ln (n)/\lambda + O(1).
\end{align*}

\section{Final Remarks}

The average
time for broadcasting an information with  one mobile object is equal to $n \HarmonicN{n-1} + n$. 
If we have two mobile objects then this time is equal to $n \HarmonicN{n} + n + \sqrt{\frac{\pi  n}{2}}+ O(1)$.
However, for a fair comparison of 
these two models we should divide the evolution
in the second model in rounds consisting of two consecutive steps --- we obtain then a speed
up (compared to model with one mobile object) with factor 2 plus small additional 
cost of order $\sqrt{\frac{\pi  n}{8}}$. 
For three mobile objects we should divide the evolution speed by $3$ and it also gives speed up compared with two mobile objects.
When we have $n$ stable and $n$ mobile objects
then we should divide the evolution speed by $n$ and still we have $2 \ln n + O(1)$.
It shows that even if we model information spreading as a process on the complete bipartite graph with $2 n$ objects, not on the complete graph with $n$ objects like in \cite{FriezeGrimmet85, Pittel87}, we still obtain asymptotically the same speed of information spreading.

\IEEEtriggeratref{4}
\bibliographystyle{IEEEtran}
\bibliography{IEEEabrv,bibliography}

\begin{thebibliography}{10}
\providecommand{\url}[1]{#1}
\csname url@samestyle\endcsname
\providecommand{\newblock}{\relax}
\providecommand{\bibinfo}[2]{#2}
\providecommand{\BIBentrySTDinterwordspacing}{\spaceskip=0pt\relax}
\providecommand{\BIBentryALTinterwordstretchfactor}{4}
\providecommand{\BIBentryALTinterwordspacing}{\spaceskip=\fontdimen2\font plus
\BIBentryALTinterwordstretchfactor\fontdimen3\font minus
  \fontdimen4\font\relax}
\providecommand{\BIBforeignlanguage}[2]{{%
\expandafter\ifx\csname l@#1\endcsname\relax
\typeout{** WARNING: IEEEtran.bst: No hyphenation pattern has been}%
\typeout{** loaded for the language `#1'. Using the pattern for}%
\typeout{** the default language instead.}%
\else
\language=\csname l@#1\endcsname
\fi
#2}}
\providecommand{\BIBdecl}{\relax}
\BIBdecl

\bibitem{FriezeGrimmet85}
A.~M. Frieze and G.~R. Grimmet, ``The shortest-path problem for graphs with
  random arc-lengths,'' \emph{Discrete Appl. Math.}, pp. 57--77, 1985.

\bibitem{Pittel87}
B.~Pittel, ``On spreading a rumor,'' \emph{SIAM J. Appl. Math.}, vol.~47,
  no.~1, pp. 213--223, 1987.

\bibitem{DBLP:conf/podc/DemersGHIL87}
A.~J. Demers, D.~H. Greene, C.~Hauser, W.~Irish, J.~Larson, S.~Shenker, H.~E.
  Sturgis, D.~C. Swinehart, and D.~B. Terry, ``Epidemic algorithms for
  replicated database maintenance,'' in \emph{PODC}, 1987, pp. 1--12.

\bibitem{DBLP:conf/focs/KarpSSV00}
R.~M. Karp, C.~Schindelhauer, S.~Shenker, and B.~V{\"o}cking, ``Randomized
  rumor spreading,'' in \emph{FOCS}, 2000, pp. 565--574.

\bibitem{DBLP:journals/rsa/FeigePRU90}
U.~Feige, D.~Peleg, P.~Raghavan, and E.~Upfal, ``Randomized broadcast in
  networks,'' \emph{Random Struct. Algorithms}, vol.~1, no.~4, pp. 447--460,
  1990.

\bibitem{DBLP:conf/spaa/Elsasser06}
R.~Els{\"a}sser, ``On the communication complexity of randomized broadcasting
  in random-like graphs,'' in \emph{SPAA}, P.~B. Gibbons and U.~Vishkin,
  Eds.\hskip 1em plus 0.5em minus 0.4em\relax ACM, 2006, pp. 148--157.

\bibitem{DBLP:conf/podc/BerenbrinkEF08}
P.~Berenbrink, R.~Els{\"a}sser, and T.~Friedetzky, ``Efficient randomised
  broadcasting in random regular networks with applications in peer-to-peer
  systems,'' in \emph{PODC}, R.~A. Bazzi and B.~Patt-Shamir, Eds.\hskip 1em
  plus 0.5em minus 0.4em\relax ACM, 2008, pp. 155--164.

\bibitem{1347167}
B.~Doerr, T.~Friedrich, and T.~Sauerwald, ``Quasirandom rumor spreading,'' in
  \emph{SODA '08: Proceedings of the nineteenth annual ACM-SIAM Symposium on
  Discrete Algorithms}.\hskip 1em plus 0.5em minus 0.4em\relax Philadelphia,
  PA, USA: Society for Industrial and Applied Mathematics, 2008, pp. 773--781.

\bibitem{ANACOMB}
P.~Flajolet and R.~Sedgewick, \emph{Analytic Combinatorics}, 1st~ed.\hskip 1em
  plus 0.5em minus 0.4em\relax Cambridge University Press, January 2008.

\bibitem{MMITZENMACHER}
M.~Mitzenmacher and E.~Upfal, \emph{Probability and Computing. Randomized
  Algorithms and Probabilistic Analysis}.\hskip 1em plus 0.5em minus
  0.4em\relax Cambridge University Press, 2005.

\end{thebibliography}
\end{document}